\documentclass[a4paper,USenglish,cleveref, autoref, thm-restate,numberwithinsect]{lipics-v2021}

\pdfoutput=1 

\title{Kernelization for Counting Problems on Graphs: Preserving the Number of Minimum Solutions}

\titlerunning{Kernelization for Counting Problems on Graphs}

\author{Bart M.\,P. Jansen}{Eindhoven University of Technology, The Netherlands}{b.m.p.jansen@tue.nl}{https://orcid.org/0000-0001-8204-1268}{\flag[3cm]{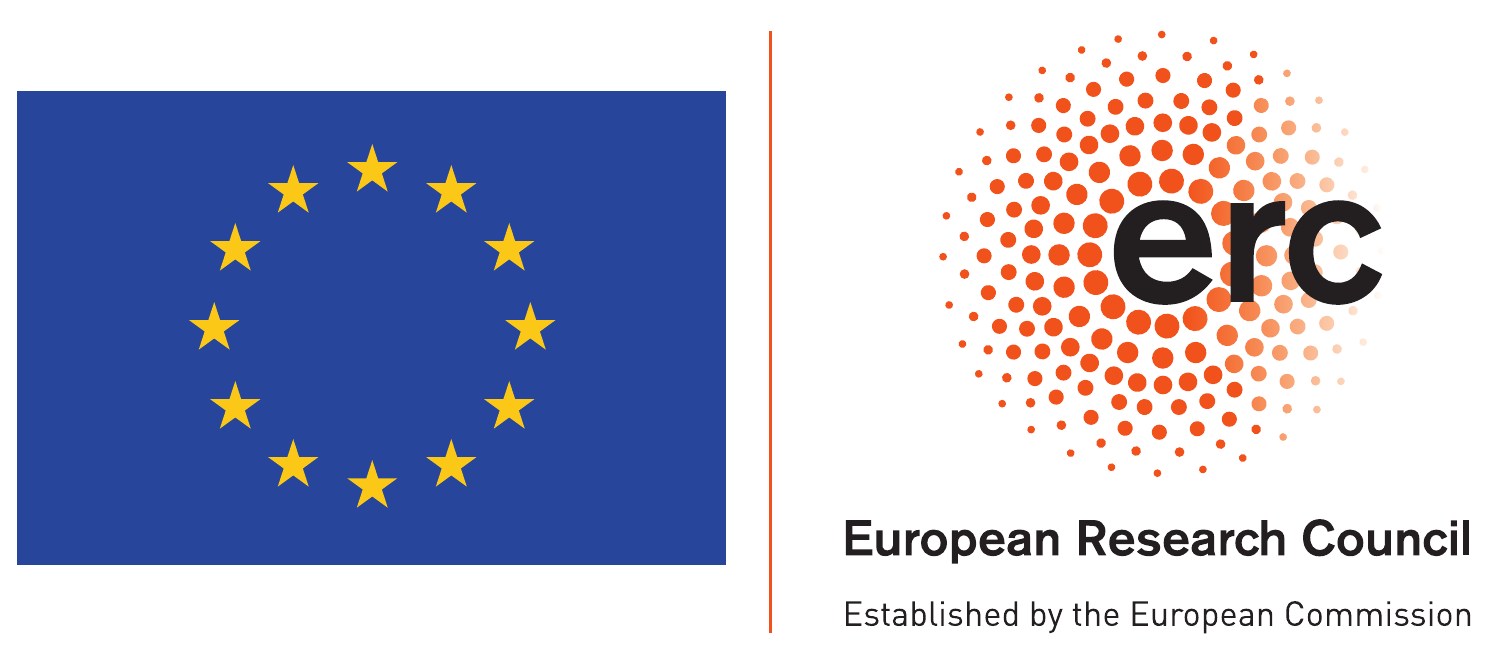}Funded by the European Research Council (ERC) under the European Union's Horizon 2020 research and innovation programme (grant agreement No 803421, ReduceSearch).}

\author{Bart van der Steenhoven}{Eindhoven University of Technology, The Netherlands}{b.j.v.d.steenhoven@student.tue.nl}{}{Supported by Project No. ICT22-029 of the Vienna Science Foundation (WWTF).}

\authorrunning{B.M.P.~Jansen and B.~v.d.~Steenhoven}

\Copyright{Bart~M.~P.~Jansen and Bart van der Steenhoven} 


\hideLIPIcs


\category{} 

\relatedversion{} 

\nolinenumbers 

\EventEditors{John Q. Open and Joan R. Access}
\EventNoEds{2}
\EventLongTitle{42nd Conference on Very Important Topics (CVIT 2016)}
\EventShortTitle{CVIT 2016}
\EventAcronym{CVIT}
\EventYear{2016}
\EventDate{December 24--27, 2016}
\EventLocation{Little Whinging, United Kingdom}
\EventLogo{}
\SeriesVolume{42}
\ArticleNo{23}

\newcommand{\bO}{\mathcal{O}}
\newcommand{\poly}{\mathrm{poly}}
\newcommand{\cmfvs}{\mathrm{\#minFVS}}
\newcommand{\cmds}{\mathrm{\#minDS}}
\newcommand{\tw}{\mathrm{tw}}
\newcommand{\bin}{\mathrm{bin}}

\RequirePackage{algorithm}
\usepackage{algpseudocodex}

\begin{document}

\maketitle

\begin{abstract}
A kernelization for a parameterized decision problem~$\mathcal{Q}$ is a polynomial-time preprocessing algorithm that reduces any parameterized instance~$(x,k)$ into an instance~$(x',k')$ whose size is bounded by a function of~$k$ alone and which has the same \textsc{yes}/\textsc{no} answer for~$\mathcal{Q}$. Such preprocessing algorithms cannot exist in the context of counting problems, when the answer to be preserved is the number of solutions, since this number can be arbitrarily large compared to~$k$. However, we show that for counting minimum feedback vertex sets of size at most~$k$, and for counting minimum dominating sets of size at most~$k$ in a planar graph, there is a polynomial-time algorithm that either outputs the answer or reduces to an instance~$(G',k')$ of size polynomial in~$k$ with the same number of minimum solutions. This shows that a meaningful theory of kernelization for counting problems is possible and opens the door for future developments. Our algorithms exploit that if the number of solutions exceeds~$2^{\poly(k)}$, the size of the input is exponential in terms of~$k$ so that the running time of a parameterized counting algorithm can be bounded by~$\poly(n)$. Otherwise, we can use gadgets that slightly increase~$k$ to represent choices among~$2^{\bO(k)}$ options by only~$\poly(k)$ vertices.
\end{abstract}

\section{Introduction}
\subparagraph{Background and motivation.} Counting problems, whose answer is an integer giving the number of objects of a certain kind rather than merely \textsc{yes} or \textsc{no}, have important applications in fields of research such as artificial intelligence~\cite{dempster-rule-orponen, approximate-reasoning-roth}, statistical physics~\cite{ising-model-jerrum, glauber-dynamics-luby, theo-physics-harary} and network science~\cite{network-motifs-ron}. They have been studied extensively in classical complexity, underpinning fundamental results such as Toda's theorem~\cite{Toda91} and the \#P-completeness of the permanent~\cite{Valiant79}. A substantial research effort has targeted the parameterized complexity of counting problems, leading to parametric complexity-notions like \#W[1]-hardness~\cite{Curticapean16,FlumG04,McCarthin06} and FPT algorithms to solve several counting problems. For example, FPT algorithms were developed to count the number of size-$k$ vertex covers~\cite{Fernau05}, or the number of occurrences of a size-$k$ pattern graph~$H$ in a host graph~$G$~\cite{CurticapeanDM17}. 

This paper is concerned with an aspect of parameterized algorithms which has been largely neglected for counting problems: that of efficient preprocessing with performance guarantees, i.e., kernelization. A kernelization for a parameterized decision problem~$\mathcal{Q}$ is a polynomial-time preprocessing algorithm that reduces any parameterized instance~$(x,k)$ into an instance~$(x',k')$ whose size is bounded by a function of~$k$ alone and which has the same \textsc{yes}/\textsc{no} answer for~$\mathcal{Q}$. Over the last decade, kernelization has developed into an important subfield of parameterized algorithms, as documented in a textbook dedicated to the topic~\cite{FominLSZ19}. Given the success of kernelization for decision problems, one may wonder: can a theory of provably-efficient preprocessing for counting problems be developed?

Consider a prototypical problem such as \textsc{Feedback Vertex Set} in undirected graphs, in which the goal is to find a small vertex set whose removal breaks all cycles. What could be an appropriate notion of counting kernelization for such a problem? The concept of efficient preprocessing towards a provably small instance with the same answer could be instantiated as follows: given a pair~$(G,k)$, the preprocessing algorithm should output a pair~$(G',k')$ whose size is bounded by a function of~$k$ such that the number of size-$k$ feedback vertex sets in~$G$ is equal to the number of size-$k'$ feedback vertex sets in~$G'$. However, this task is clearly impossible. Given a graph consisting of a length-$n$ cycle with parameter~$k=1$, the number of solutions is~$n$ which can be arbitrarily large compared to~$k$, while for any reduced instance~$(G',k')$ of size bounded in~$k$, the number of solutions can be at most~$2^{|V(G')|} \leq f(k)$. Without allowing the size of the reduced instance to depend on~$n$, it seems that preprocessing while preserving the answer to the counting problem is impossible. 

Over the years, there have been two approaches to deal with this obstacle\footnote{A third~\cite{LokshtanovMSZ23} approach was announced shortly before this paper went to print; it allows a polynomial-time lifting step to compute the number of solutions to the original instance from the number of solutions to the reduced instance.}. Thurley~\cite{Thurley07}, inspired by concepts in earlier work~\cite{NishimuraRT05}, proposed a notion of counting kernelization which effectively reduces a counting problem to an enumeration problem. He considered problems such as \textsc{Vertex Cover} and \textsc{$d$-Hitting Set}. In his framework, the preprocessing algorithm has to output an instance of size bounded by a function of~$k$, in such a way that for any solution to the reduced instance, we can efficiently determine to how many solutions of the original instance it corresponds. Hence by enumerating \emph{all} solutions on the reduced instance, we can obtain the number of solutions to the original instance. A significant drawback of this approach therefore lies in the fact that to solve the counting problem on the original instance, we have to enumerate all solutions on the reduced instance. Since counting can potentially be done much faster than enumeration, it is not clear that this preprocessing step is always beneficial.

A second notion for counting kernelization was proposed by Kim, Selma, and Thilikos~\cite{KimST18,Thilikos21}. Their framework (which also applies to \textsc{Feedback Vertex Set}) considers two types of algorithms: a \emph{condenser} that maps an input instance~$(G,k)$ to an instance~$(G',k')$ of an auxiliary \emph{annotated} problem involving weights on the vertices of~$G'$, and an \emph{extractor} that recovers (typically not in polynomial time) the number of solutions to~$(G,k)$ from the weighted instance~$(G',k')$. The number of vertices of~$G'$ is required to be bounded in~$k$, but the weights are allowed to be arbitrarily large, thereby sidestepping the issue described above. This means that in terms of the total encoding size, the weighted graph~$(G',k')$ is not guaranteed to be smaller than~$(G,k)$ and in general the total number of bits needed to encode the weighted graph cannot be bounded by a function of~$k$ alone. The condenser-extractor framework has the same drawback as the framework by Thurley: a standard counting problem is reduced to a more complicated type of problem, in this case one involving weights and annotations.

The goal of this paper is to show that there is an alternative way to overcome the obstacle for counting kernelization, which leads to a notion of preprocessing in which the problem to be solved on the reduced instance is of exactly the same nature as the original. Our solution is inspired by the typical behavior of kernelization algorithms for decision problems: we formalize the option of already finding the answer during the preprocessing phase. Note that many algorithms, such as the famous Buss~\cite{BussG93} kernelization for \textsc{Vertex Cover}, work by applying reduction rules to arrive at the reduced instance, or discover the \textsc{yes}/\textsc{no} answer to the decision problem during preprocessing. Our kernelization algorithms for counting problems will have the same behavior: they will either reduce to a $\poly(k)$-sized instance of the same problem that has exactly the same answer to the counting problem, or they outright answer the counting problem during their polynomial-time computation. To our initial surprise, such preprocessing algorithms exist for several classic problems.


\subparagraph{Our results.} To begin the exploration of this new type of counting kernelization, we revisit two prominent graph problems: \textsc{Feedback Vertex Set} in general undirected graphs and \textsc{Dominating Set} in planar graphs. The decision versions of these problems (does graph~$G$ have a solution of size at most~$k$?) have kernels with $\bO(k^2)$~\cite{Iwata17,Thomasse10} and~$\bO(k)$ vertices~\cite{AlberFN04,BodlaenderFLPST16,FominLST20}, respectively. We consider the problem of counting the number of \emph{minimum-size} solutions, parameterized by the size~$k$ of a minimum solution. (We discuss counting inclusion-minimal solutions in the conclusion.)  For a graph~$G$ and integer~$k$, we denote by $\cmfvs(G, k)$ the number of minimum feedback vertex sets in~$G$ of size at most~$k$ in~$G$. Hence~$\cmfvs(G, k)$ is equal to~$0$ if the feedback vertex number of~$G$ exceeds~$k$, and otherwise is equal to the number of minimum solutions. The analogous concept for minimum dominating sets is denoted~$\cmds(G, k)$. Our result for \textsc{Feedback Vertex Set} reads as follows.

\begin{restatable}{theorem}{fvsKernelTheorem}\label{thm:cfvs}
There is a polynomial-time algorithm that, given a graph $G$ and integer $k$, either
\begin{itemize}
    \item outputs $\cmfvs(G, k)$, or
    \item outputs a graph $G'$ and integer $k'$ such that $\cmfvs(G, k) = \cmfvs(G', k')$ and $|V(G')| = \bO(k^5)$ and $k' = \bO(k^5)$.
\end{itemize}
\end{restatable}

\vspace{0.2cm}

For \textsc{Dominating Set} on planar graphs, we give an analogous algorithm that either outputs~$\cmds(G,k)$ or reduces to a \emph{planar} instance~$(G',k')$ with~$|V(G')|, k' = \bO(k^3)$ such that~$\cmds(G,k)=\cmds(G',k')$. Hence if the parameter is small, the task of counting the number of minimum solutions can efficiently be reduced to the \emph{same} counting task on a provably small instance.

The high-level approach is the same for both problems. We use insights from existing kernels for the decision version of the problem to reduce an input instance~$(G,k)$ into one~$(G',k')$ with the same number of minimum solutions, such that~$G'$ can be decomposed into a ``small'' core together with~$\poly(k)$ ``simply structured but potentially large'' parts. For \textsc{Dominating Set}, this takes the form of a protrusion decomposition; for \textsc{Feedback Vertex Set} the decomposition is more elementary. Then we consider two cases. If~$|V(G)| > 2^k$, we employ an FPT algorithm running in time~$2^{\bO(k)} \cdot \poly(n)$ to count the number of minimum solutions and output it. Since~$n > 2^k$, this step runs in polynomial time. If~$|V(G)| \leq 2^k$, then we show that each of the~$\poly(k)$ simply structured parts can be replaced with a gadget of size~$\poly(k)$ without affecting the number of minimum solutions. In this step, we typically increase the size of minimum solutions slightly to allow a small vertex set to encode exponentially many potential solutions. For example, an instance of \textsc{Feedback Vertex Set} consisting of a cycle of length~$2^{10}$ (which has~$2^{10}$ different optimal solutions), can be reduced to the graph consisting of~$10$ pairs~$(a_i, b_i)$, each pair connected by two parallel edges. The latter graph also has~$2^{10}$ minimum solutions, each of size~$10$. To carry out this approach, the most technical part is to show how to decompose the input instance into parts in which it is easy to analyze how many different choices an optimal solution can make.

\subparagraph{Organization.} The remainder of the paper is structured as follows. After presenting preliminaries in~\cref{sec:preliminaries}, we illustrate our approach for \textsc{Feedback Vertex Set} in \cref{sec:cfvs}. The more technical application to \textsc{Dominating Set} on planar graphs is given in \cref{sec:cpds}. We conclude in \cref{sec:conclusion} with a reflection on the potential of this approach to counting kernelization. 


\section{Preliminaries} \label{sec:preliminaries}

\subsection{Graphs}
All graphs we consider are undirected; they may have parallel edges but no self-loops. A graph $G$ therefore consists of a set $V(G)$ of vertices and a multiset $E(G)$ of edges of the form $\{u,v\}$ for distinct $u,v \in V(G)$.
For a vertex $v \in V(G)$, we refer to the \emph{open neighborhood} of $v$ in $G$ as $N_G(v)$ and to the \emph{closed neighborhood} of $v$ as $N_G[v]$. For a set of vertices $X \subseteq V(G)$, the open and closed neighborhoods are defined as $N_G(X) = (\bigcup_{v \in X} N_G(v)) \setminus X$ and $N_G[X] = \bigcup_{v \in X} N_G[v]$.
The degree of vertex $v$ in graph $G$, denoted by $\deg_G(v)$, is equal to the number of edges incident to $v$ in $G$. We refer to the subgraph of $G$ induced by a vertex set $X \subseteq V(G)$ as $G[X]$. We use $G-X$ as a way to write $G[V(G) \setminus X]$ and $G - v$ as a shorthand for $G - \{v\}$. A graph $H$ is a \emph{minor} of $G$ if $H$ can be formed by contracting edges of a subgraph of $G$. 

A graph is \emph{planar} if it can be embedded in the plane in such a way that its edges intersect only at their endpoints. Such an embedding is called a \emph{planar embedding} of $G$. We make use of the following two properties of planar graphs. 

\begin{theorem}[Wagner's theorem] \label{thm:wagner}
A graph $G$ is planar if and only if $G$ contains neither $K_5$ nor $K_{3,3}$ as a minor. 
\end{theorem}

\begin{lemma} [{\cite[Lemma 13.3]{kernel-book-fomin}}] \label{lem:planar-deg-3-bound}
Let $G$ be a planar graph, $C \subseteq V(G)$, and let $N_3$ be a set of vertices from $V(G) \setminus C$ such that every vertex from $N_3$ has at least three neighbors in $C$. Then, $|N_3| \leq \max \{ 0, 2|C|-4 \}$.      
\end{lemma}

A \emph{feedback vertex set} of a graph $G$ is a set $S \subseteq V(G)$ such that $G-S$ is a forest, i.e., acyclic. The \emph{feedback vertex number} of a graph is the size of a smallest feedback vertex set of that graph. 
A \emph{dominating set} of a graph $G$ is a set $D \subseteq V(G)$ such that $N_G[D] = V(G)$. The \emph{domination number} of a graph is the size of a smallest dominating set of that graph. We say that a set $X \subseteq V(G)$ dominates $U \subseteq V(G)$ if $U \subseteq N_G[X]$.

We define $V_{\neq 2}(G)$ to be the set of vertices of graph $G$ that do not have degree two. 
We refer to a \emph{chain} $C$ of $G$ as a connected component of $G-V_{\neq 2}(G)$. We say that a chain $C$ is a \emph{proper chain} if $N_G(C) \neq \emptyset$ and we then refer to $N_G(C)$ as the \emph{endpoints} of $C$. 

\subsection{Tree decompositions}

For completeness we give the definition of a tree decomposition, which is needed to formally define the notion of protrusion below.

\begin{definition}[Tree decomposition]
A tree decomposition of a graph $G$ is a pair $(T, \chi)$, where $T$ is a rooted tree and $\chi$ assigns a bag $\chi(w) \subseteq V(G)$ to every $w \in V(T)$ such that the following conditions are satisfied.
\begin{itemize}
    \item $\bigcup_{w \in V(T)} \chi(w) = V(G)$.
    \item For all $\{u,v\} \in E(G)$, there exists a node $w \in V(T)$ such that $\{u,v\} \subseteq \chi(w)$.
    \item For every $v \in V(G)$, the set $\{w \in V(T) \mid  v \in \chi(w) \}$ induces a non-empty connected subtree of $T$.
\end{itemize}
\end{definition}
The \emph{width} of a tree decomposition is $\max_{w\in V(T)} |\chi(w)| - 1$. The \emph{treewidth} of a graph $G$, denoted by $\tw(G)$, is the minimum possible width of a tree decomposition of $G$. For a node $w$ of $T$, we define $T_w$ to be the subtree of $T$ rooted at $w$. For a subtree $T'$ of $T$, we define $\chi(T') := \bigcup_{z \in V(T')} \chi(z)$.

\subsection{Protrusion decompositions}
The definitions concerning protrusion decompositions are based on those from \cite[\S 15]{kernel-book-fomin}.

\begin{definition}[Boundary]
 For a graph $G$ and a vertex subset $X \subseteq V(G)$, we refer to the boundary of $X$ as $\partial(X) = \{v \in X \mid N_G(v) \setminus X \neq \emptyset\}$.
\end{definition}

\begin{definition}[Protrusion]
For an integer $t >0$, a $t$-protrusion in a graph $G$ is a vertex set $X \subseteq V(G)$ such that $\tw(G[X]) \leq t$ and $|\partial(X)| \leq t$.
\end{definition}

\begin{definition}[Protrusion decomposition]
For integers $\alpha$, $\beta$ and $t$, an $(\alpha, \beta, t)$-protrusion decomposition of a graph $G$ is a tree decomposition $(T, \chi)$ of $G$ such that the following conditions are satisfied. 
\begin{itemize}
    \item The tree $T$ is a rooted tree with root $r$ and $|\chi(r)| \leq \alpha$.
    \item For every $v \in V(T)$ except $r$, we have $|\chi(v)| \leq t$.
    \item The root $r$ has degree at most $\beta$ in $T$.
\end{itemize}
\end{definition}
Following this definition, observe that for each child $w$ of $r$ in $T$, we have that $\chi(T_w)$ is a $t$-protrusion in $G$. For such a protrusion, we observe that $\partial(\chi(T_w)) \subseteq \chi(w)$ and we additionally define $\partial^*(\chi(T_w)) = \chi(w)$ to be the \emph{extended boundary} of protrusion $\chi(T_w)$.

\begin{definition}[Treewidth modulator]
    A set $S \subseteq V(G)$ is a treewidth-$\eta$-modulator in a graph $G$ if $\tw(G-S) \leq \eta$.
\end{definition}

\begin{lemma} \label{lem:mod-to-protrusiondecomp}
There is a polynomial-time algorithm that, given a planar graph $G$ and a treewidth-$\eta$-modulator $S$ in $G$, computes an
$(\alpha, \beta, t)$-protrusion decomposition $(T, \chi)$ such that the following conditions are satisfied.
\begin{itemize}
    \item For the root $r$ of $T$ we have $S \subseteq \chi(r)$.
    \item For each child $w$ of $r$ in $T$, the set $\chi(T_w)$ is a $t$-protrusion in $G$ such that $\chi(T_w) \cap S \subseteq \partial^*(\chi(T_w))$ and ${|\partial^*(\chi(T_w)) \cap S| \leq 2}$.
    \item $\alpha = \bO(\eta |S|)$, $\beta = \bO(\eta |S|)$ and $t = 3\eta + 2$.
\end{itemize}
\end{lemma}
The lemma above corresponds with Lemma 15.14 of \cite{kernel-book-fomin} with the addition of the second property. This property follows from the proofs of Lemma 15.13 and Lemma 15.14 in \cite{kernel-book-fomin}.

\section{Counting feedback vertex sets} \label{sec:cfvs}
In this section, we explain the technique that allows us to either count the number of minimum feedback vertex sets of a graph $G$ in polynomial time, or reduce $G$ to a provably small instance with the same number of minimum solutions. We start by showing that, by using a few reduction rules, we can already reduce $G$ to an equivalent instance with a specific structure. This reduction is based on the $\bO(k^3)$-vertex kernel for the decision \textsc{Feedback Vertex Set} problem presented by Jansen~\cite{fvs-kernel-ppt-jansen}. We choose to use this kernel over the better-known and smaller-size kernels by Thomassé~\cite{Thomasse10} and Iwata~\cite{Iwata17} because those rely on multiple reduction rules that are not safe for counting minimum solutions. 

As is common for \textsc{Feedback Vertex Set}, we consider the graph we are working with to be undirected and we allow parallel edges. 
In this section, we make use of two reduction rules which are common for kernels of the decision variant of \textsc{Feedback Vertex Set}.
\begin{description}
    \item[(R1)] If there is an edge of multiplicity larger than two, reduce its multiplicity to two.
    \item[(R2)] If there is a vertex $v$ with degree at most one, remove $v$.
\end{description}

It can easily be verified that if an instance $(G,k)$ is reduced to $(G',k')$ by one of the rules above, we have $\cmfvs(G,k) = \cmfvs(G',k')$. Hence, these rules are safe in the context of counting minimum feedback vertex sets. Observe that if (R2) has exhaustively been applied on a graph $G$, then all vertices in $V_{\neq2}(G)$ have degree at least three.
For our purposes, we will need one more method to reduce the graph. We first present a lemma that motivates this third reduction rule.
\begin{lemma} \label{lem:cfvs-degsum}
Let $X$ be a (not necessarily minimum) feedback vertex set of a graph $G$ that is reduced with respect to (R2) and let $\mathcal{C}$ be the set of connected components of $G - V_{\neq2}(G)$. Then:
\begin{alphaenumerate}
    \item \label{enum:degsum-a} $|V_{\neq2}(G)| \leq |X| + \sum_{v \in X} \deg_G(v)$, and
    \item \label{enum:degsum-b} $|\mathcal{C}| \leq |X| + 2 \sum_{v \in X} \deg_G(v)$.
\end{alphaenumerate}
\end{lemma}
\begin{proof}
Consider the forest $F := G - X$. Partition the vertices of $V_{\neq2}(G) \cap V(F) = V_{\neq2}(G) \setminus X$ into sets $V'_{\leq 1}$, $V'_2$ and $V'_{\geq 3}$ for vertices that have respectively degree at most one, degree two or degree at least three in $F$. Furthermore, let $V_\ell$ denote the leaf nodes of $F$ that have degree two in $G$.
Since $G$ has no vertices of degree at most one by (R2), the leaves of $F$ are exactly the vertices $V_\ell \cup V'_{\leq 1}$.
In any tree, the number of vertices of degree at least three is less than the number of leaves, thus 
 $|V'_{\geq 3}| \leq |V'_{\leq 1}| + |V_\ell|$. Each vertex in $V'_2$ has at least one edge to $X$ since they have degree at least three in $G$ and degree exactly two in $F$. For a similar reason, each vertex in $V'_{\leq1}$ has at least two edges to $X$. Each vertex in $V_\ell$ has one edge to $X$.
Putting this together gives the following inequality, from which (\ref{enum:degsum-a}) directly follows.
\begin{align} \label{ali:degsum}
    \sum_{v \in X} \deg_G(v) \geq |V'_2| + 2|V'_{\leq1}| + |V_\ell| \geq |V'_2| + |V'_{\leq1}| + |V'_{\geq 3}| = |V_{\neq2}(G) \setminus X|
\end{align} 

To bound the size of the set $\mathcal{C}$ of connected components of $G - V_{\neq 2}(G)$, we instead bound the size of the set $\mathcal{C'}$ of connected components of $G - V_{\neq 2}(G) - X$. Observe that $|\mathcal{C}| \leq |\mathcal{C'}| + |X|$ since removing a vertex from a graph reduces the number of connected components by at most one. (Such a removal can \emph{increase} the number of components by an arbitrary number, which is irrelevant for our argument.) Since $X$ is an FVS, the connected components in $\mathcal{C'}$ can be seen as proper chains. Chains in $\mathcal{C'}$ that have both endpoints in $F$ act as edges between those endpoints when it comes to the connectivity of $F$. This means that, since $F$ is a forest, there can be at most $|V_{\neq2}(G) \cap V(F)| \leq \sum_{v \in X} \deg_G(v)$ of such chains by \autoref{ali:degsum}. All other chains will have at least one endpoint in $X$, which means there can be at most $\sum_{v \in X} \deg_G(v)$ of them, implying (\ref{enum:degsum-b}).
\end{proof}

Based on \autoref{lem:cfvs-degsum}, the goal of the third reduction rule is to decrease the degree of the vertices of a feedback vertex set. This idea is captured in \autoref{lem:cfvs-degred}. After presenting this lemma, we combine these results in \autoref{lem:cfvs-kernel} to create an algorithm to reduce a graph $G$ to an, in context of counting minimum feedback vertex sets, equivalent graph with a bounded number of vertices of degree other than two and a bounded number of chains.  
\begin{lemma} \label{lem:cfvs-degred}
There exists a polynomial-time algorithm that, given a graph $G$ reduced with respect to (R1), an integer $k$, a vertex $v \in V(G)$ and a feedback vertex set $Y_v \subseteq V(G) \setminus \{ v \}$ of $G$, outputs a graph $G'$ obtained by removing edges from $G$ such that $\deg_{G'}(v) \leq |Y_v| \cdot (k+4)$ and $\cmfvs(G, k) = \cmfvs(G', k)$.
\end{lemma}
\begin{proof}
Consider forest $F := G - (Y_v \cup \{v \})$. For each $u \in Y_v$, mark trees of $F$ that have an edge to both $v$ and $u$ until either all such trees are marked or at least $k+2$ of them are marked. Then, we construct a graph $G'$ from $G$ by removing all edges between $v$ and trees of $F$ that were not marked. 

We shall first prove the bound on the degree of $v$ in $G'$. The vertex $v$ can have edges to vertices in $Y_v$ or in $F$. Since (R1) has been exhaustively applied, there can be at most $2|Y_v|$ edges between $v$ and $Y_v$. Each tree in $F$ has at most one edge to $v$, since otherwise $Y_v$ would not be an FVS of $G$. In $G'$, only trees that were marked still have an edge to $v$, and since we mark at most $k+2$ trees per vertex in $Y_v$, we have at most $|Y_v| \cdot (k+2)$ of such trees. Combining this gives $\deg_{G'}(v) \leq |Y_v| \cdot (k+4)$. 

To show that $\cmfvs(G, k) = \cmfvs(G', k)$, we prove that a vertex set $X \subseteq V(G)$ with $|X| \leq k$ is an FVS of $G$ if and only if it is an FVS of $G'$. Clearly any FVS of $G$ is an FVS of $G'$ since $G'$ is constructed from $G$ by deleting edges. For the opposite direction, assume that $X$ is an FVS of $G'$ and assume for a contradiction that $X$ is not an FVS of $G$. Then $G-X$ has a cycle $W$. This cycle must contain an edge $\{v, w\}$ of $E(G) \setminus E(G')$ since $G' - X$ is acyclic. By construction of $G'$, we know that $w$ is a vertex that belongs to an unmarked tree $T$. Since cycle $W$ intersects $T$ and 
since $T$ is a connected component of $G - (Y_v \cup \{v\})$, there must be a vertex $u \in Y_v$ that has an edge to $T$. 
Since the edge between $v$ and $T$ is removed in $G'$, there exist $k+2$ other trees that are marked and have an edge to both $v$ and $u$. Since $|X| \leq k$, at least two of these trees are not hit by $X$. Furthermore, since $v$ and $u$ are part of $W$, they are also not contained in $X$. Therefore, there is a cycle in $G'-X$ through $v$, $u$ and two of the aforementioned trees, contradicting that $X$ is an FVS of $G'$.
\end{proof}

\begin{lemma} \label{lem:cfvs-kernel}
There is a polynomial-time algorithm that, given a graph $G$ and integer $k$, outputs a graph $G'$ and integer $k'$ such that the following properties are satisfied.
\begin{itemize}
    \item $\cmfvs(G, k) = \cmfvs(G', k')$.
    \item $k' \leq k$.
    \item $|V_{\neq 2}(G')|= \bO(k^3)$.
    \item $G'-V_{\neq 2}(G')$ has $\bO(k^3)$ connected components.
\end{itemize}
\end{lemma}
\begin{proof}
In our approach, we make use of the linear-time 4-approximation algorithm by Bar-Yehuda et al.~\cite{fvs-approx-bar-yehuda} that can also approximate the more general problem of: for a given graph, find the smallest FVS that does not contain a given vertex. Our first step is to exhaustively apply (R1) on $G$ and compute a 4-approximate FVS $X$ of the graph. If $|X| > 4k$ then the feedback vertex number of $G$ is larger than $k$, so we can return a trivial, constant size instance $G'$ and $k'$ such that $\cmfvs(G', k') = 0$. Otherwise, let $G'$ and $k'$ be a copy of $G$ and $k$. For each vertex $v \in X$, compute a 4-approximate FVS $Y_v$ of $G'$ that does not contain $v$. If $|Y_v| > 4k$, then there does not exist a solution of size at most $k$ that does not contain $v$, so remove $v$ from $G'$ and reduce $k'$ by one. Otherwise, use \autoref{lem:cfvs-degred} to reduce the degree of $v$ in $G'$. Finally, we exhaustively apply reduction rule (R2) on $G'$.

Computing the 4-approximations and applying the reduction rules can be done in polynomial time. Each rule can only be applied a polynomial number of times, thus the algorithm runs in polynomial time.
The fact that $\cmfvs(G, k) = \cmfvs(G', k')$ follows from safety of the reduction rules used in the algorithm and $k' \leq k$ follows from the fact that $k'$ is never increased. We know that $|V_{\neq2}(G')| = \bO(k^3)$ and that $G'-V_{\neq2}(G')$ has $\bO(k^3)$ connected components due to \autoref{lem:cfvs-degsum} and the following bound: 
\[ \sum_{v \in X} \deg_{G'}(v) \leq \sum_{v \in X} |Y_v| \cdot (k+4) \leq \sum_{v \in X} 4k \cdot (k+4) = |X| \cdot 4k \cdot (k+4) = \bO(k^3). \qedhere \] 
\end{proof}

The result of \autoref{lem:cfvs-kernel} is in and of itself not a proper kernel yet, since the chains of the graph it produces can be of arbitrary length. Our strategy to address this is as follows. If these chains are large in terms of $k$, then we can run an FPT algorithm in $\poly(n)$ time to count the number of minimum solutions. Otherwise, the chains can be replaced by structures of size $\poly(k)$ that do not change the number of minimum feedback vertex sets the instance has. This approach is captured in the following two lemmas and combined in the proof of \autoref{thm:cfvs}.

\begin{lemma} \label{lem:cfvs-chainred}
There exists a polynomial-time algorithm that, given a graph $G$ with a chain $C$ and an integer $k$, outputs a graph $G'$ obtained from $G$ by replacing $C$ with a vertex set $C'$, and an integer $k'$, such that:
\begin{itemize}
    \item $\cmfvs(G, k) = \cmfvs(G', k')$,
    \item $G - C = G' - C'$,
    \item $N_G(C) = N_{G'}(C')$,
    \item $|C'| = \bO(\log (|C|)^2)$, and
    \item $k' = k+ \bO(\log (|C|)^2)$.
\end{itemize}
\end{lemma}
\begin{proof}
In case $C$ is a proper chain, the endpoints are defined as $N_G(C)$. If $C$ is not a proper chain, which happens if $C$ is a cycle in $G$, we choose an arbitrary vertex of $C$ to act as its endpoint. For simplicity, we consider $C$ to have two endpoints, where in some cases these two endpoints might be the same vertex. 

First, we assume that the number of vertices of the chain, not including its endpoints, is a power of two, i.e. $|C| = 2^p$ for some integer $p$. Let $v,u \in V(G)$ be the endpoints of $C$ (possibly $v = u$) and let $C = \{c_0, c_1, \cdots, c_{2^p -1}\}$. Then we construct our graph $G'$ by replacing $C$ by a gadget $C'$, of which an example can be seen in \autoref{fig:fvs-gadget}. It consists of the following elements.
\begin{itemize}
    \item A vertex $w$ with edges to both $v$ and $u$.
    \item Pairs of vertices $a_i, b_i$ for $0 \leq i < p$ such that there is an edge of multiplicity two between $w$ and $a_i$ and between $a_i$ and $b_i$.
\end{itemize}
Additionally, we set $k' = k + p$. 

\begin{figure}
            \centering%
            \subcaptionbox{\label{fig:fvs-chain}}
                {\includegraphics[page=1]{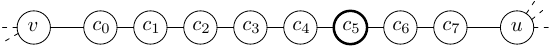}} \\ \vspace{0.5cm}
            \subcaptionbox{\label{fig:fvs-replacement}}
                {\includegraphics[page=2]{figures/minimumfvs-gadget-alt-again.pdf}}
            \caption{(\subref{fig:fvs-chain}) A chain structure of size eight with two endpoints. (\subref{fig:fvs-replacement}) The replacement of the structure. An example of the mapping $f$ is also illustrated through the vertices with a thicker border.}
            \label{fig:fvs-gadget}
\end{figure}

We shall now prove that $\cmfvs(G, k) = \cmfvs(G', k')$. To this end, we define a mapping $f$ from the set of minimum feedback vertex sets of $G$ to those of $G'$ and show that this is a bijection, which immediately implies that the two sets have the same cardinality. 
For a natural number $m$, define $\bin(m)$ to be the binary representation of $m$ on $p$ bits and define $\bin(m)_i$ to be the $i$'th least significant bit of $\bin(m)$.
\[ f(X) = \begin{cases}
    X \cup \{a_i \mid 0 \leq i < p\} & \mbox{if } X \cap C = \emptyset \\ \\
    (X \setminus C) \cup \{w\} \cup \{a_i \mid 0 \leq i < p \wedge \bin(m)_i = 0\} \\ \cup  \{b_i \mid 0 \leq i < p \wedge \bin(m)_i = 1\}  & \mbox{if } X \cap C = \{c_m\} \\
\end{cases}\]

As a first observation, note that for any minimum feedback vertex set $X$ of $G$, we have $|X \cap C| \leq 1$ as picking any one vertex from $C$ will already break all cycles that go through the chain. 
\begin{claim} \label{claim:fvs_fx_is}
    If a set $X$ is a minimum FVS of $G$, then $f(X)$ is a minimum FVS of $G'$. 
\end{claim}
\begin{claimproof}
We prove this in two parts. First we prove that $f(X)$ is an FVS of $G'$ and then we prove that $f(X)$ is indeed an FVS of $G'$ of minimum size. 

To prove that $f(X)$ is an FVS of $G'$, we use a proof by contradiction. Assume $f(X)$ is not an FVS of $G'$, in which case a (simple) cycle $W$ exists in $G'-f(X)$. This cycle must intersect $C'$ as $X \setminus C = f(X) \setminus C'$ and $G-C = G'-C'$ thus $G-C-X = G'-C'-f(X)$, which means $W$ would otherwise also exist in $G-X$, contradicting that $X$ is an FVS of $G$. The cycle $W$ cannot contain a $b_j$ vertex since by definition of $f$, for $0 \leq i < p$, either $a_i$ or $b_i$ is in $f(X)$, which would either mean $b_j$ is isolated in $G'-f(X)$ or is removed. Also, $W$ can not contain any $a_j$ vertex as, by definition of $f$, if $a_j$ is not in $f(X)$, then both its neighbors $w$ and $b_j$ are in $f(X)$. This leaves only the option that $W$ intersects $C'$ through only vertex $w$. This means that $W-\{w\}$ contains a path from $v$ to $u$ in $G'-f(X) - \{w\}$. However, since as mentioned before $G-C-X = G'-C'-f(X)$, this same path also exists in $G-X$. Furthermore, since $w \notin f(X)$, that must mean that $X \cap C = \emptyset$ so $(W \setminus \{w\})\cup C$ would form a cycle in $G-X$ contradicting that $X$ is an FVS of $G$. 

Next we prove that $f(X)$ is a minimum FVS of $G'$. 
Assume for sake of a contradiction that $f(X)$ is not a minimum FVS of $G'$ due to the existence of a $Y \subseteq V(G')$ with $|Y| < |f(X)|$ such that $Y$ is an FVS of $G'$. We first observe that any FVS of $G'$ contains at least $p$ vertices from $C'$ since all $p$ of the pairs $(a_i, b_i)$ form vertex disjoint cycles. Similarly so, if an FVS of $G'$ contains $w$, then it contains at least $p+1$ vertices from $C'$. We distinguish two cases.
\begin{description}
    \item[Case $w \notin Y$:] Then $Y \setminus C'$ is an FVS of $G$. If $G- (Y\setminus C')$ would contain a cycle, then there would also exist a cycle in $G'-Y$ since $G-C = G'-C'$ and both graphs $G- (Y\setminus C')$ and $G'-Y$ have a $vu$ path, the former through $C$ and the latter through $w$. Furthermore, $|Y \setminus C'| \leq |Y| - p < |f(X)| - p = |X| + p - p = |X|$, contradicting that $X$ is a minimum FVS of $G$. 
    \item[Case $w \in Y$:] Then $X' := (Y \setminus C') \cup \{c_j\}$ is an FVS of $G$ for any $0 \leq j < 2^p$. If this were not the case, then, since $X'$ contains a chain vertex, a cycle would need to exist completely in $G - C - X'$ which is the same graph as $G' - C' - Y$. This would contradict $Y$ being an FVS of $G'$. Furthermore, $|X'| \leq |Y| - (p + 1) + 1 < |f(X)| - p = |X| + p - p = |X|$, contradicting that $X$ is a minimum FVS of $G$. 
\end{description} 
As both cases lead to a contradiction, we conclude that $f(X)$ is a minimum FVS of $G'$.
\end{claimproof}

\begin{claim}
    The function $f$ is bijective. 
\end{claim}
\begin{claimproof}
We first argue that $f$ is injective. Let $X$ and $X'$ be two minimum feedback vertex sets of $G$ such that $X \neq X'$. That means $X \setminus C \neq X' \setminus C$ or $X \cap C \neq X' \cap C$. The former immediately allows us to conclude that $f(X) \neq f(X')$ since $X \setminus C = f(X) \setminus C'$ and $X' \setminus C = f(X') \setminus C'$. In the second case, we have two options. The first is that $X \cap C = \{c_j\}$ and $ X' \cap C = \{c_m\} $ for some $j \neq m$, and since binary representations are unique, they lead to different sets $f(X)$ and $f(X')$. 
The second option is that either $X$ or $X'$ contains no vertex from $C$ while the other one does. Then only one of $f(X)$ or $f(X')$ contains $w$ and the other one does not, so $f(X) \neq f(X')$.

It remains to show that $f$ is surjective. To this end, we take an arbitrary minimum FVS $Y$ of $G'$ and show that there exists a minimum FVS $X$ of $G$ such that $Y = f(X)$. We distinguish two cases:

\begin{description}
    \item[Case $w \notin Y$:] For this case, first observe that $\{ a_i \mid 0 \leq i < p \} \subseteq Y$ since otherwise $w$ would form a cycle with one of these $a_i$ vertices in $G-Y$. Furthermore, $b_i \notin Y$ for $0 \leq i < p$ since $Y$ already contains $a_i$, leaving $b_i$ isolated in $G-Y$ and thus a redundant choice for a minimum FVS. From this we can conclude that $Y \cap C' = \{ a_i \mid 0 \leq i < p \}$. Therefore, taking $X := Y \setminus C'$ would give $f(X) = Y$. We have already argued in case $w \notin Y$, that $Y \setminus C'$ is an FVS of $G$. To show that it is a minimum FVS of $G$, note that if there was an $X' \subseteq V(G)$, $|X'| < |X|$ such that $X'$ is an FVS of $G$, then $f(X')$ is an FVS of $G'$ and $|f(X')| = |X'| + p < |X| + p = |Y \setminus C'| + p = |Y| - p + p = |Y|$ which would contradict $Y$ being a minimum FVS of $G'$.   
    \item[Case $w \in Y$:] For this case, we instead observe that for each pair $\{a_i, b_i\}$, exactly one of $\{ a_i, b_i\}$ is in $Y$, as otherwise $Y$ would not be an FVS of minimum size. Since there are $p$ of such pairs, there is a unique value $0 \leq j < 2^p$ such that the binary representation of $j$ corresponds to the choice of $a$ and $b$ vertices in $Y$. We can then choose $X := (Y \setminus C') \cup \{c_j \}$. We clearly have that $f(X) = Y$. Also, we have already shown before that $X$ is an FVS of $G$ and the argument that $X$ is a minimum FVS of $G$ is analogous to that in case $w \notin Y$. 
\end{description}
From this reasoning we can conclude that $f$ is a bijective function from the set of minimum feedback vertex sets of $G$ to the set of minimum feedback vertex sets of $G$.
\end{claimproof} 

In the argument above, we assumed that the length of the chain is a power of two. We can address this by noting that any natural number can be written as a sum of unique powers of two. Similarly, we can decompose the chain $C$ into a number of subpaths each having a number of vertices that is a unique power of two. We can then apply the replacement described above on each subpath individually to get multiple replacement structures in a chain between the endpoints of $C$.  
As seen before, the size of the replacement is linear in the exponent of the length of the chain. In the worst case, when expressing $|C|$ in binary as a sum of distinct powers of two, the exponents of these powers sum up to $\bO(\log(|C|)^2)$, which is also the bound on the number of vertices used in our replacement and on the increase in the parameter value. 
\end{proof}

We remark for \autoref{lem:cfvs-chainred} that a similar chain replacement gadget without parallel edges can be constructed, at the expense of a linear increase in the size of the gadget.

By adapting the iterative compression algorithm by Cao et al.~\cite{fvs-new-measure-cao} for the decision \textsc{Feedback Vertex Set} problem, we can derive the following lemma. Its proof is given in \autoref{app:fpt-cfvs-algo}. 
\begin{restatable}{lemma}{cfvsFPTAlgo} \label{lem:cfvs-algo}
There is an algorithm that, given a graph $G$ and integer $k$, computes $\cmfvs(G, k)$ in $2^{\bO(k)} \cdot \poly (n)$ time.
\end{restatable}

\fvsKernelTheorem*
\begin{proof}
First we use the algorithm from \autoref{lem:cfvs-kernel} to find a graph $G^*$ and integer $k^* \leq k$ such that $\cmfvs(G, k) = \cmfvs(G^*, k^*)$, $|V_{\neq 2}(G^*)| = \bO(k^3)$ and $G^* - V_{\neq 2}(G^*)$ has $\bO(k^3)$ connected components (chains). Then, if there is a chain of size larger than $2^k$, we can run the algorithm from \autoref{lem:cfvs-algo} to compute $\cmfvs(G^*, k^*)$ in $\poly (n)$ time since $n > 2^k$. Otherwise, all chains of $G^*$ have size at most $2^k$ and we can use \autoref{lem:cfvs-chainred} on $G^*$ to find a graph $G'$ and integer $k'$ such that:
\begin{itemize}
    \item $\cmfvs(G', k') = \cmfvs(G^*, k^*) = \cmfvs(G, k)$,
    \item $|V(G')| = \bO(k^3) + \bO( k^3 \cdot \log (2^k)^2) = \bO(k^5)$, and
    \item $k' = \bO(k + k^3 \cdot \log (2^k)^2) = \bO(k^5)$. \qedhere
\end{itemize}
\end{proof}

\section{Counting dominating sets in planar graphs} \label{sec:cpds}

In this section, we show how our approach can also be used with a different technique. In \autoref{sec:cfvs} we adapted a decision kernel and replaced remaining, potentially large structures by smaller, similarly behaving structures. In this section we will instead focus on protrusion decompositions and the reduction of protrusions. 

We consider the problem of counting minimum dominating sets of a planar graph. 
In the context of dominating sets, we assume all graphs to be undirected and simple, so they do not have parallel edges. We start by showing that if a planar graph $G$ has a dominating set of size at most $k$, we can use this to find a protrusion decomposition. 

\begin{lemma} \label{lem:cpds-protrusiondecomp}
Let $G$ be a planar graph and $k$ an integer. If $G$ has a dominating set of size at most $k$, then $G$ has an $(\bO(k), \bO(k), 8)$-protrusion decomposition $(T, \chi)$ with root $r$ of $T$ such that, for each child $w$ of $r$, the set of vertices $\chi(T_w) \setminus \partial(\chi(T_w))$ can be dominated by at most two vertices in $\partial^*(\chi(T_w)) = \chi(w)$.
Moreover, there is a polynomial-time algorithm that, given $G$ and $k$, either finds such a protrusion decomposition or determines that the domination number of $G$ is greater than $k$. 
\end{lemma}
\begin{proof}
The decision version of \textsc{Dominating Set} on planar graphs admits a polynomial-time approximation scheme~\cite{approx-algos-planar-baker}. As a first step, we can use this to find a 2-approximate dominating set $S$ of $G$ in polynomial time. If $|S| > 2k$ then the domination number of $G$ is greater than $k$, so we are done. Otherwise, in case $G$ is connected, we can turn $S$ into a connected dominating set $S'$ of $G$ such that $|S'| \leq 3|S| \leq 6k$ \cite[Claim 15.10]{kernel-book-fomin}. If $G$ is not connected, we can do the same for each connected component of the graph. Since $G$ is a planar graph, $G-S'$ is an outerplanar graph and thus has treewidth at most two \cite[Lemma 78]{arboretum-bodlaender}. This means that $S'$ is a treewidth-2-modulator for $G$.
Therefore, we can use \autoref{lem:mod-to-protrusiondecomp} to find in polynomial time an $(\bO(k), \bO(k), 8)$-protrusion decomposition that satisfies the described requirements.
\end{proof}

We call a vertex set $C \subseteq V(G)$ a \emph{wide diamond} if there are distinct vertices $v,u \in V(G) \setminus C$ such that for all $c \in C$, we have that $N_G(c) = \{v,u\}$. We consider $v$ and $u$ to be the endpoints of the wide diamond $C$. We make the following observation concerning minimum dominating sets on wide diamonds.
\begin{observation} \label{obs:wide-diamond}
    If $C$ is a wide diamond with distinct endpoints $v,u$ in a graph $G$ and $|C| \geq 3$, then any minimum dominating set $X$ of $G$ satisfies the following properties.
    \begin{enumerate}
    \item \label{obsit:wide-diamond-1} $X \cap \{v,u\} \neq \emptyset$, since if $X$ contains neither $v$ nor $u$, we must have that $C \subseteq X$, while $(X \setminus C) \cup \{v,u\}$ would then be a dominating set of smaller size.
    \item \label{obsit:wide-diamond-2} If $\{v, u \} \subseteq X$, then $X \cap C = \emptyset$, since $N_G[c] \subseteq N_G[\{v,u\}]$ for all $c \in C$.
    \item \label{obsit:wide-diamond-3} $|X \cap C| \leq 1$, since by Property~\ref{obsit:wide-diamond-1}, we can assume w.l.o.g. that $v \in X$, and $N_G[c] \setminus N_G[v] = \{u\}$ for all $c \in C$.
    \end{enumerate}
\end{observation}

In \autoref{lem:cpds-netred} we show how we can efficiently replace a wide diamond by a smaller, similarly behaving structure. We show how to use this replacement to reduce the protrusions of a protrusion decomposition in \autoref{lem:cpds-protrusionred}.

\begin{lemma} \label{lem:cpds-netred}
There exists a polynomial-time algorithm that, given a planar graph $G$ with a wide diamond $C$ and an integer $k$, outputs a planar graph $G'$ obtained from $G$ by replacing $C$ with a vertex set $C'$, and an integer $k'$, such that:
\begin{itemize}
    \item $\cmds(G, k) = \cmds(G', k')$,
    \item $G - C = G' - C'$,
    \item $N_G(C) = N_{G'}(C')$,
    \item $|C'| = \bO(\log (|C|)^2)$, and
    \item $k' = k+ \bO(\log (|C|)^2)$.
\end{itemize}
\end{lemma}
\begin{proof}
Let vertices $v$ and $u$ be the endpoints of the wide diamond $C$. Assume without loss of generality that $C$ consists of at least five vertices, since the statement is trivial otherwise. For simplicity, we first assume that $|C| = 2^p + 3$ for some integer $p$, so let $C = \{c_0, c_1, \cdots, c_{2^p+2}\}$. We then construct $G'$ from $G$ by replacing $C$ by $C'$, which has the following structure. 
In $C'$, we leave three vertices $c_0, c_1, c_2$ from $C$ untouched and replace all others. 
\begin{itemize}
    \item Add vertex $x_v$ with an edge to $v$ and vertex $x_u$ with an edge to $u$. 
    \item For each $0 \leq i < p$, add two adjacent vertices $a_i$ and $b_i$ and make these vertices adjacent to both $x_v$ and $x_u$.
    \item Add a degree-one vertex $e_i$ adjacent to $b_i$.
\end{itemize}
It is easy to see that $G'$ is planar since the constructed gadget has a planar embedding with its only boundary vertices $x_v, x_u$ on a single face (see \autoref{fig:pds-gadget}). The gadget can be drawn into any face of an embedding of $G$ that contains both $u$ and $v$, whose existence is guaranteed because $N_G(c_0) = \{v, u\}$. 
For the parameter, we set $k' = k + p$. 

\begin{figure}
            \centering%
            \subcaptionbox{\label{fig:pds-net}}
                {\includegraphics[page=1]{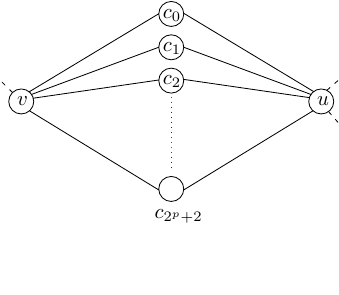}}\hfill
            \subcaptionbox{\label{fig:pds-replacement}}
                {\includegraphics[page=2]{figures/minimumpds-gadget-alt.pdf}}
            \caption{(\subref{fig:pds-net}) A wide diamond structure with endpoints. (\subref{fig:pds-replacement}) The replacement of the structure.}
            \label{fig:pds-gadget}
\end{figure}

We define a function $f$ from the set of minimum dominating sets of $G$ to the set of minimum dominating sets of $G'$ and we prove that $\cmds(G,k) = \cmds(G',k')$ by showing that $f$ is a well-defined, bijective function. 
\[ f(X) = \begin{cases}
    X \cup \{b_i \mid 0 \leq i < p\} & \mbox{if } X \cap C \subseteq \{c_0, c_1, c_2\}\\ \\
    (X \setminus C) \cup  \{x_u\} \cup \{b_i \mid 0 \leq i < p \wedge \bin(m-3)_i = 0\} \\ \cup \{e_i \mid 0 \leq i < p \wedge \bin(m-3)_i = 1\}  & \mbox{if } X \cap C = \{c_m\} \wedge m \geq 3 \wedge v \in X  \\ \\
    (X \setminus C) \cup  \{x_v\} \cup \{b_i \mid 0 \leq i < p \wedge \bin(m-3)_i = 0\} \\ \cup \{e_i \mid 0 \leq i < p \wedge \bin(m-3)_i = 1\}  & \mbox{if } X \cap C = \{c_m\} \wedge m \geq 3 \wedge u \in X 
\end{cases}\]

Note that since we assumed $|C| \geq 5$, by \autoref{obs:wide-diamond} the function $f$ covers all cases for a minimum DS $X$ of $G$. Also note that $\{c_0, c_1, c_2\}$ still forms a wide diamond of size at least three between $v$ and $u$, so the properties of the observation still apply in $G'$.

\begin{claim} \label{claim:ds-fx-is-min}
If a set $X$ is a minimum DS of $G$, then $f(X)$ is a 
minimum DS of $G'$.
\end{claim}
\begin{claimproof}
We prove this claim in two parts. We first prove that $f(X)$ is a DS of $G'$, and then we prove that it is a DS of minimum size. 

To prove that $f(X)$ is a DS of $G'$, we use a proof by contradiction, so assume $f(X)$ is not a DS of $G'$. This means there is a vertex $w \in V(G')$ that is not dominated by $f(X)$. 
If $w$ is a vertex of $G' - C' - \{v,u\}$, then since $G' - C' = G - C$, we have that $w$ is also a vertex of $G - C - \{v,u\}$. Furthermore, since $X \setminus C = f(X) \setminus C'$ and $N_G(C) = \{v, u\}$, if $w$ is not dominated by $f(X)$ it is also not dominated by $X$, which contradicts $X$ being a dominating set of $G$.
Thus $w \in C'$ or $w \in \{v,u\}$. The former can never be true, as in all cases the function $f$ takes the union over a dominating set of $C' \setminus \{c_0, c_1, c_2\}$ and by Property~\ref{obsit:wide-diamond-1} of \autoref{obs:wide-diamond}, either $v$ or $u$ is in $X$, and thus also in $f(X)$, ensuring that $c_0, c_1$ and $c_2$ are dominated. So, assume without loss of generality that $w = u$. Then, by Property~\ref{obsit:wide-diamond-1} of \autoref{obs:wide-diamond}, that means $v \in X$ and since $u$ is dominated by $X$ in $G$ that $X \cap C = \{c_m\}$ for some $c_m$. However, by definition of $f$, that would mean that either $c_m \in f(X)$ or $x_u \in f(X)$ and both $c_m$ and $x_u$ have an edge to $u$ which contradicts that $u$ is not dominated by $f(X)$. Thus $f(X)$ is a dominating set of $G'$. 

We also use a proof by contradiction to show that $f(X)$ is a minimum DS of $G'$. Assume that $f(X)$ is not a minimum DS of $G'$ due to the existence of a set $Y \subseteq V(G')$ with $|Y| < |f(X)|$ such that $Y$ is a dominating set of $G'$. Assume that $Y$ is a minimum such dominating set. 
We distinguish three cases:
\begin{description}
    \item[Case $x_v, x_u \notin Y$:] We can then show that $\{b_i \mid 0 \leq i < p\} \subseteq Y$. If there would be a $b_j \notin Y$, then since also $x_v, x_u \notin Y$, we would need to have that $a_j, e_j \in Y$ since these are the only remaining options to dominate these vertices. However, since $N_{G'}[\{a_j, e_j\}] = N_{G'}[b_j]$, taking $(Y \setminus \{a_j, e_j\}) \cup \{b_j\}$ would result a dominating set of smaller cardinality, contradicting that $Y$ is a minimum dominating set of $G'$. Thus $\{b_i \mid 0 \leq i < p\} \subseteq Y$. From this it follows that $Y' := Y \setminus \{b_i \mid 0 \leq i < p\}$ is a dominating set of $G$. We already have $v\in Y$ or $u \in Y$ by Property~\ref{obsit:wide-diamond-1} of \autoref{obs:wide-diamond}, ensuring that all vertices in $C$ are dominated. 
    Also, since $x_v, x_u \notin Y$, the set $Y$ dominates $u$ and $v$ from $V(G) \setminus (C\setminus\{c_0,c_1,c_2\})$. 
    Furthermore, $|Y'| = |Y| - p < |f(X)| - p = |X| + p - p = |X|$, contradicting that $X$ is a minimum DS of $G$. 
    \item[Case $x_v \in Y$:] As a first step, we can realize that $v \notin Y$. If $v$ would be in $Y$, then we could construct a DS of $G'$ of smaller size than $Y$ by taking $(Y \setminus C') \cup \{b_i \mid 0 \leq i < p\}$ (potentially also including $x_u$ if $x_u \in Y$). Thus, $v \notin Y$ which by Property~\ref{obsit:wide-diamond-1} of \autoref{obs:wide-diamond} implies that $u \in Y$.
    Knowing this, we consider $Y' := (Y \setminus C') \cup \{c_j\}$ for some $0 \leq j < 2^p+3$. This set $Y'$ is a dominating set of $G$ since all vertices $G-C-\{v,u\}$ are unaffected by the change, $u \in Y'$ which ensures all vertices $c_i$ are dominated, and $c_j \in Y'$ ensuring that $v$ is dominated. Notice that $|Y \cap C'| \geq p + 1$ since $x_v \in Y$ and $Y$ must contain either $e_i$ or $b_i$ for each $0 \leq i < 2^p$ in order to dominate $e_i$.
    Thus, $|Y'| \leq |Y| - (p+1) + 1 = |Y| - p < |f(X)| - p = |X| + p - p = |X|$, contradicting that $X$ is a minimum DS of $G$.
    \item[Case $x_u \in Y$:] The argument for this case is symmetric to that of the previous case. 
\end{description}
In all cases we reach a contradiction from which we can conclude that $f(X)$ is a minimum DS of $G'$.
\end{claimproof}
From this we can conclude that $f$ is a well-defined mapping from the set of minimum dominating sets of $G$ to those of $G'$. 

\begin{claim}
    The function $f$ is bijective. 
\end{claim}
\begin{claimproof}
The function $f$ is clearly injective as when $X \neq X'$, either $X \setminus C \neq X' \setminus C$ or $X \cap C \neq X' \cap C$ which both imply that $f(X) \neq f(X')$ by definition of $f$. 
Following the same type of argumentation as in \autoref{claim:ds-fx-is-min}, it is easy to verify the following. For any minimum dominating set $Y'$ in $G'$, if we define $Y$ as follows:
\[ Y = \begin{cases}
     Y' \setminus \{b_i \mid 0 \leq i < p\} & \mbox{if } Y' \cap \{x_v, x_u\} = \emptyset\\ 
     (Y' \setminus C') \cup \{c_j\}  & \mbox{if } Y' \cap \{x_v, x_u\} \neq \emptyset
\end{cases}\]
then $Y$ is a minimum DS of $G$ and $f(Y) = Y'$ for the correct choice of $j$, which is determined by which $b$ and $e$ vertices are in $Y'$. Thus the function is surjective, and therefore also a bijection.
\end{claimproof}

Similar as for counting feedback vertex sets, we can address the issue that we assumed $|C| = 2^p+3$ for some $p$ by decomposing the wide diamond $C$ into multiple smaller wide diamonds such that one consists of three vertices and all others consist of a number of vertices that is a unique power of two.
In order for the complete structure to get the behavior we desire, it suffices to not replace the one wide diamond consisting of three vertices to let this act as $c_0, c_1, c_2$ in the replacement procedure for all other replacements. We then replace all other wide diamonds entirely in the way described above. 
This allows us to replace any wide diamond $C$ by a structure with $\bO(\log (|C|)^2)$ vertices, and $k' = k +  \bO (\log (|C|)^2)$. 
\end{proof}

\begin{lemma} \label{lem:cpds-protrusionred}
There is a polynomial-time algorithm that, given a planar graph $G$, an integer $k$ and a protrusion $P \subseteq V(G)$ 
such that $P = \chi(T_w)$ for some $(\alpha, \beta, 8)$-protrusion decomposition $(T, \chi)$ with root $r$ and child of root $w$ and such that $P \setminus \partial(P)$ can be dominated by a set of at most two vertices in $\partial^*(P)$, 
outputs a planar graph $G'$ obtained from $G$ by replacing $P$ with a vertex set $P'$, and integer $k'$, such that:
\begin{itemize}
    \item $\cmds(G, k) = \cmds(G',k')$,
    \item $G - (P \setminus \partial^*(P)) = G' - (P' \setminus \partial^*(P))$,
    \item $N_G(P) = N_{G'}(P')$,
    \item $|P'| = \bO(\log (|P|)^2)$, and
    \item $k' = k + \bO( \log(|P|)^2)$.
\end{itemize}
\end{lemma}
\begin{proof}
The proof is structured as follows. We will first remove certain edges and vertices in the protrusion that are not relevant in the context of dominating sets. We then prove that the remaining protrusion can only have $\bO(1)$ vertices that are not part of a wide diamond. Finally, we use \autoref{lem:cpds-netred} to replace the wide diamond structures that appear in this remaining protrusion.

Consider a minimum dominating set $S$ of $G$. The first thing we note is that $S$ contains at most $8$ vertices from $P$. This is because $\partial^*(P)$ is a dominating set of $P$ with $|\partial^*(P)| \leq 8$ and for all $D \subseteq P$, we have that $N_G[D] \subseteq N_G[\partial^*(P)]$. Thus, if $S$ would contain more than $8$ vertices of $P$, then $(S \setminus P) \cup \partial^*(P)$ would be a dominating set of $G$ of smaller size, contradicting that $S$ is a minimum dominating set. Additionally, we note that the vertex set $S \cap P$ dominates $P \setminus \partial(P)$, which follows directly from the definition of a boundary. We shall refer to a set $D \subseteq P$ as a \emph{partial dominating set} (PDS) of $P$ if $D$ dominates $P \setminus \partial(P)$. Based on these observations, the goal of our reduction is to preserve these partial dominating sets of size at most 8. 

The first step will be to remove redundant edges between vertices of $P \setminus \partial^*(P)$. More formally, an edge $\{v,u\}$ with $v,u \in P \setminus \partial^*(P)$ is called \emph{redundant} if for each partial dominating set $D$ of $P$ of size at most 8, the set $D$ is still a partial dominating set of $P$ in graph $G$ when we remove edge $\{v,u\}$. We can check if an edge is redundant by enumerating the set of partial dominating sets of $P$ of size at most 8 in $\poly (n)$ time and checking if removing the edge changes this set. Next, we remove certain vertices from $P \setminus \partial^*(P)$ using the following observation. If there is a vertex $v \in V(G)$ with more than two neighbors of degree one, then all minimum dominating sets of $G$ will contain $v$ and none of those degree-one neighbors. Therefore, removing all but two of the degree-one neighbors does not affect the set of minimum dominating sets of $G$. After we have removed redundant edges from $P$, we use this observation to remove unnecessary degree-one vertices from $P \setminus \partial^*(P)$. 

Assume now that we used the techniques described above to remove edges and vertices from $P$. We can analyze the structure that the remaining graph $G[P]$ has by counting the types of vertices that appear in it. We first consider one specific case.

If $\partial^*(P) = \{s\}$ for some vertex $s$, then we can always find a replacement $P'$ for $P$ that has only a constant number of vertices. For this, we first realize that in this specific case, any minimum dominating set $S$ of $G$ contains at most one vertex from $P$. Therefore, $S$ can only contain a vertex from $P$ that dominates all vertices in $P \setminus \partial(P)$. If $P \setminus \partial(P) = \emptyset$ then $P$ already has constant size, so assume that this is not the case.
Let $U$ be the set of vertices of $P \setminus \partial^*(P)$ that individually form a partial dominating set of $P$.
If $U = \emptyset$, then, since $P \setminus \partial(P) \neq \emptyset$, any minimum dominating set of $G$ must contain $s$. Thus we can let $P'$ consist of only $s$ and two new vertices, both having only an edge to $s$.
If $U \neq \emptyset$, then we construct $P'$ from $P$ by removing all vertices other than $s$ and those in $U$. There can be at most three vertices in $U$, since otherwise $G[U \cup \{s\}]$ would form $K_5$, which would contradict $G$ being planar by \autoref{thm:wagner}. To see why this reduction is safe, note that for any minimum dominating set $S$ of $G$, we must have that $S \cap P = \{v\}$ for some vertex $v \in U \cup \{s\}$. If $|S \cap P| > 1$, then $S$ is not of minimum size, and if $v \notin U \cup \{s\}$, then there is a vertex in $P \setminus \partial(P)$ not dominated by $S$.

From this point on, assume $|\partial^*(P)| > 1$ and let $s_1, s_2$ be two distinct vertices in $\partial^*(P)$ that together dominate $P \setminus \partial(P)$.
We partition the vertices of $P \setminus \partial^*(P)$ into sets $P_{\leq 1}$, $P_2$ and $P_{\geq 3}$ based on whether their degree in $G[P]$ is at most one, is two, or is at least three, respectively. We first look at $P_{\leq 1}$. 
Note that there cannot exist vertices of degree zero since $\{s_1, s_2\}$ is a partial dominating set of $P$, thus all vertices of $P \setminus \partial^*(P)$ must have an edge to at least one of these. Similarly, for vertices of degree one, the one edge incident to them must have as other endpoint either $s_1$ or $s_2$. There cannot be more than two degree-one neighbors per vertex, since then we would have removed them previously. From this we can conclude that $|P_{\leq 1}| \leq 4$. 

Before we look at $P_{\geq 3}$, we will first prove the following claim to be correct. This claim will in turn be useful for bounding the number of vertices contained in $P_{\geq 3}$.

\begin{claim} \label{claim:deg3-nb-bound}
If $P$ has no redundant edges, then for each distinct pair of vertices $z_1,z_2 \in P$, there can be at most $96$ vertices in $(N_{G[P]}(z_1) \cap N_{G[P]}(z_2)) \setminus \partial(P)$ that are adjacent to at least one vertex in $P \setminus (\partial(P) \cup \{z_1,z_2\})$. 
\end{claim}
\begin{claimproof}
Let $R \subseteq (N_{G[P]}(z_1) \cap N_{G[P]}(z_2) )\setminus \partial(P)$ be the set of vertices that are adjacent to both $z_1$ and $z_2$ and at least one vertex in $P \setminus (\partial(P) \cup \{z_1,z_2\})$. Assume for a contradiction that $|R| > 96$. We first prove that all partial dominating sets of size at most 8 of $P$ must contain at least one of $\{z_1,z_2\}$. If this were not the case, then there exists a partial dominating set $S$ with $|S| \leq 8$ such that $z_1, z_2 \notin S$. All vertices in $R$ must be dominated by $S$, so there is a vertex $v \in S$ that dominates more than $96/8 = 12$ vertices of $R$. This means that $v$ has an edge to more than $11$ vertices in $R$. However, then $v$, $z_1$, $z_2$ and three of those vertices of $R$ adjacent to $v$ would together form $K_{3,3}$, which contradicts $G$ being planar by \autoref{thm:wagner}. Thus, each PDS of size at most $8$ must contain $z_1$ or $z_2$.

We associate a vertex from $P \setminus (\partial(P) \cup \{z_1,z_2\})$ to each vertex of $R$ through a function $f \colon R \to P \setminus (\partial(P) \cup \{z_1,z_2\})$, where for each vertex $v \in R$, we let $f(v)$ be an arbitrary vertex in $N_{G[P]}(v) \setminus (\partial(P) \cup \{z_1, z_2\})$. Next, let $f(R) := \{f(v) \mid v \in R\}$ be the set of vertices of $P \setminus (\partial(P) \cup \{z_1,z_2\})$ that are associated to some vertex of $R$. See \autoref{fig:Q2-example} for an example of the vertices involved in the structure under consideration. 
No vertex in $f(R)$ has an edge to both $z_1$ and $z_2$, since then all edges from $R$ to that vertex would have been redundant due to $z_1$ or $z_2$ being in every partial dominating set of size at most 8. 
We then claim that no vertex $v \in V(G) \setminus \{z_1, z_2\}$ is adjacent to three vertices of $f(R)$. For a proof by contradiction, assume that $v$ is adjacent to three distinct vertices $f(a), f(b), f(c) \in f(R)$. Then this would contradict $G$ being planar due to it containing $K_{3,3}$ as a minor. One side of the minor consists of $\{z_1, z_2, v\}$, and the other side consists of the vertices obtained by contracting the edges $\{a, f(a)\}$, $\{b, f(b)\}$ and $\{c, f(c)\}$. 

Furthermore, each vertex in $f(R)$ can be the image of at most two vertices of $R$. If there were distinct $a,b,c \in R$ such that $f(a) = f(b) = f(c)$, then $\{z_1, z_2, f(a)\}$ and $\{a,b,c\}$ would together form $K_{3,3}$. Thus, since $|R| > 96$, we must have $|f(R)| > 48$.
Each vertex in $f(R)$ is non-adjacent to at least one of $\{z_1, z_2\}$, so at least half of them are non-adjacent to the same one. Assume without loss of generality that at least half of them, so more than $24$, are non-adjacent to $z_1$. Let $U$ refer to this set of vertices.
Then $z_2$ must be in all partial dominating sets of $P$ of size at most 8. If this were not the case, then there would need to be a vertex $v \in V(G) \setminus \{z_1, z_2\}$ that dominates more than $24/8 = 3$ of these vertices in $U \subseteq f(R)$, which we already argued to be impossible. Thus $z_2$ is contained in all partial dominating sets of $P$ of size at most 8. This leaves two options. Either $z_2$ has no edge to any vertex of $U$, in which case no partial dominating set of size at most $8$ exists, which is a contradiction as is witnessed by  the partial dominating set $\partial^*(P)$. As a second option, $z_2$ has an edge to at least one vertex $f(u) \in U$, which would imply that the edge between $u$ and $f(u)$ is redundant, which is also a contradiction. As all cases lead to a contradiction, we must have that $|R| \leq 96$.
\end{claimproof} 

Having proven this claim, we shall now consider $P_{\geq 3}$, which we partition further into sets $Q_1$, $Q_2$ and $Q_{\geq 3}$ based on whether they have one, two, or at least three edges to $\partial^*(P)$, respectively. We will prove that $P_{\geq 3}$ has constant size by proving that all three of these sets have constant size. We can immediately conclude that $|Q_{\geq 3}| = \bO(1)$ using \autoref{lem:planar-deg-3-bound} since $G[P]$ is planar, each vertex in $Q_{\geq 3}$ has at least three neighbors in $\partial^*(P)$ by definition, and $|\partial^*(P)| \leq 8$. 

\begin{figure}
            \centering%
            \subcaptionbox{\hfill \label{fig:Q2-example}}
                {\includegraphics{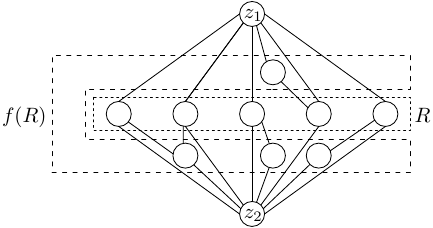}}
                \hfill
            \subcaptionbox{\hfill \label{fig:Q1-example}}
                {\includegraphics{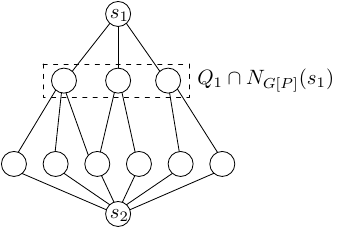}}
            \caption{(\subref{fig:Q2-example}) An example of the types of vertices involved in the proof of \autoref{claim:deg3-nb-bound}. (\subref{fig:Q1-example}) An example of the types of vertices involved in the proof of \autoref{claim:q1size}. }
            \label{fig:prot-replacement-examples}
\end{figure}

\begin{claim} \label{claim:q2size}
    The set $Q_2$ consists of at most $28 \cdot 96$ vertices.
\end{claim}
\begin{claimproof}
The set $Q_2$ consists of the vertices of $P \setminus \partial^*(P)$ that have degree at least three in $G[P]$ and have exactly two edges to $\partial^*(P)$. Thus, all vertices of $Q_2$ exist between a pair of $\partial^*(P)$ vertices. Consider an arbitrary pair of distinct vertices $z_1, z_2 \in \partial^*(P)$. Then, since by definition each vertex in $Q_2$ has an edge to at least one vertex not in $\partial^*(P)$, we can use \autoref{claim:deg3-nb-bound} to conclude that $|Q_2 \cap N_{G[P]}(z_1) \cap N_{G[P]}(z_2)| \leq 96$. Since there are at most $\binom{8}{2} = 28$ pairs of $\partial^*(P)$ vertices, we arrive at the bound stated in the claim.
\end{claimproof}

\begin{claim} \label{claim:q1size}
    The set $Q_1$ consists of at most $2 \cdot 8 \cdot 96$ vertices. 
\end{claim}
\begin{claimproof}    
The set $Q_1$ consists of the vertices of $P \setminus \partial^* (P)$ that have degree at least three in $G[P]$ and have exactly one edge to $\partial^*(P)$. Note that since $\{s_1, s_2\}$ dominates the vertices in $P \setminus \partial(P)$ (and thus also $P \setminus \partial^*(P)$), all vertices in $Q_1$ have an edge to either $s_1$ or $s_2$ and no other vertices in $\partial^*(P)$. To bound the number of vertices in $Q_1$, we will bound the number of such vertices that can be adjacent to $s_1$, and realize that that same bound applies for $s_2$. First observe that if all partial dominating sets of size at most 8 contain both $s_1$ and $s_2$, then all edges between vertices of $P \setminus \partial^*(P)$ would have been redundant and we would immediately have that $|Q_1| = 0$, so assume that this is not the case. 
We shall prove that $|Q_1 \cap N_{G[P]}(s_1)| \leq 8 \cdot 96$ using a proof by contradiction, so assume $|Q_1 \cap N_{G[P]}(s_1)| > 8 \cdot 96$. We distinguish two cases.
\begin{description}
    \item[Case:] There exists a partial dominating set of size at most 8 that does not contain $s_1$. That means that all vertices in $N_{G[P]}(s_1) \setminus \partial(P)$ can be dominated with at most 8 vertices other than $s_1$. Since $|Q_1 \cap N_{G[P]}(s_1)| > 8 \cdot 96$ there must be a vertex $v \in P \setminus \partial^*(P)$ that has an edge to more than $96$ vertices in $Q_1 \cap N_{G[P]}(s_1)$. Since the vertices in $Q_1$ have degree at least three, they must all have an edge to at least one vertex in $P \setminus (\partial^*(P) \cup \{v\})$. This means we can use \autoref{claim:deg3-nb-bound} to conclude that $|Q_1 \cap N_{G[P]}(s_1) \cap N_{G[P]}(v)| \leq 96$, which contradicts our assumption.
    \item[Case:] All partial dominating sets of $P$ of size at most 8 contain $s_1$. Note that this means that there is a partial dominating set $S$ with $|S| \leq 8$ that does not contain $s_2$, since we assumed that not all partial dominating sets of size at most 8 contain both $s_1$ and $s_2$. In this case, the vertices in $Q_1 \cap N_{G[P]}(s_1)$ all have exclusively edges to vertices that do not have an edge to $s_1$ but do have an edge to $s_2$, since otherwise the edge would be redundant or $\{s_1, s_2\}$ would not dominate the vertex. See \autoref{fig:Q1-example} for an example of the vertices involved in such a structure.
    We can find that $|N_{G[P]}(s_2)| \geq |Q_1 \cap N_{G[P]}(s_1)| > 8 \cdot 96$. This follows from the fact that each vertex in $Q_1 \cap N_{G[P]}(s_1)$ has at least two edges to vertices in $N_{G[P]}(s_2)$ and that each vertex $v \in N_{G[P]}(s_2)$ can have an edge to at most two vertices in $Q_1 \cap N_{G[P]}(s_1)$. If $v$ has edges to three vertices in $Q_1 \cap N_{G[P]}(s_1)$, then $v, s_1, s_2$ and those three vertices would form $K_{3,3}$ when we contract the edges between $s_2$ and $N_{G[P]}(s_2) \setminus \{v\}$. 
    In a similar vein, any vertex in $S$ can have an edge to at most two of these vertices in $N_{G[P]}(Q_1 \cap N_{G[P]}(s_1)) \cap N_{G[P]}(s_2)$ due to planarity and thus can dominate at most three of them. However, there are more than $8\cdot 3 = 24$ of such vertices which contradicts $S$ being a partial dominating set of $P$. 
\end{description}
In both cases we reach a contradiction, thus we must have that $|Q_1 \cap N_{G[P]}(s_1)| \leq 8 \cdot 96$. The same bound holds for the vertices of $Q_1$ that are adjacent to $s_2$ through analogous reasoning, and since $Q_1 \subseteq N_{G[P]}(s_1) \cup N_{G[P]}(s_2)$, we must have that ${|Q_1| \leq 2 \cdot 8 \cdot 96 = \bO(1)}$.
\end{claimproof}

Summarizing so far, we have shown that $|P_1| = \bO(1)$ and $|P_{\geq3}| = |Q_1| + |Q_2| + |Q_{\geq 3}| = \bO(1)$. This only leaves $P_2$, the vertices of degree two. All such vertices must have at least one edge to $\partial^*(P)$ since it is a dominating set for $P$. There are a few options for the other edge. It could have an edge to another vertex in $P_2$ which in turn has an edge to some vertex in $\partial^*(P)$. Let $P_2^*$ be the set of $P_2$ vertices involved in such a structure. If there were three or more of such structures between a pair of vertices in $\partial^*(P)$, then the edge between the degree two vertices would have been redundant. Thus, since $|\partial^*(P)| = \bO(1)$, also $|P_2^*| = \bO(1)$. Finally, consider the vertices of $P_2 \setminus P_2^*$. These are vertices of $P$ of degree two with both neighbors in the set $\partial^*(P) \cup P_{\geq3}$. The vertices of $P_2 \setminus P_2^*$ can be partitioned into $\bO(1)$ wide diamonds, namely the wide diamonds of the form $N_{G[P]}(x) \cap N_{G[P]}(y) \cap P_2$ for $x,y \in \partial^*(P) \cup P_{\geq3}$. Since $|\partial^*(P) \cup P_{\geq3}| =\bO(1)$, we also end up with $\bO(1)$ wide diamonds.  We can use \autoref{lem:cpds-netred} for each wide diamond to get a replacement $P'$ and $k'$ that satisfy the requirements stated in the lemma. 
\end{proof}

Having proven that we can reduce the size of these protrusions, we proceed similarly to as for counting feedback vertex sets. We argue that there is an algorithm that solves $\cmds$ in \autoref{lem:cpds-algo} and combine all our results in \autoref{thm:cpds}.

\begin{lemma} \label{lem:cpds-algo}
    There is an algorithm that, given a planar graph $G$ and integer $k$, computes $\cmds(G,k)$ in $2^{\bO(\sqrt{k})} \cdot \poly(n)$ time.
\end{lemma}
\begin{proof}[Proof sketch]
There exists an algorithm that solves the decision version of \textsc{Dominating Set} in $4^{\tw(G)} \cdot \tw(G)^{\bO(1)} \cdot n$ time \cite{course-book}. This algorithm can straightforwardly be adapted to also count the number of minimum dominating sets without increasing the running time. Furthermore, note that if $G$ has a dominating set of size at most $k$, then $\tw(G) = \bO(\sqrt{k})$~\cite{kernel-book-fomin}. As a final ingredient, we use that there exists a polynomial-time exact algorithm for branchwidth on planar graphs \cite{ratcatcher-branchwidth-seymour}, and that any branch decomposition of width $k$ can be turned into a tree decomposition of width $\bO(k)$ in polynomial time \cite{branch-tree-decomp-hicks}. Putting this together results in an algorithm that can compute $\cmds(G,k)$ in the mentioned time. 
\end{proof}

\begin{theorem} \label{thm:cpds}
There is a polynomial-time algorithm that, given a planar graph $G$ and integer $k$, either
\begin{itemize}
    \item outputs $\cmds(G, k)$, or
    \item outputs a planar graph $G'$ and integer $k'$ such that $\cmds(G,k) = \cmds(G',k')$ and $|V(G')| = \bO(k^3)$ and $k' = \bO(k^3)$.
\end{itemize}
\end{theorem}
\begin{proof}
First, we can use the algorithm from \autoref{lem:cpds-protrusiondecomp} on $G$ and $k$. If it reports that the domination number of $G$ is greater than k, then $\cmds(G,k) = 0$. Else, we find an $(\bO(k), \bO(k), 8)$-protrusion decomposition of $G$. If there is a protrusion that consists of more than $2^k$ vertices, we know that $n > 2^k$ so we can run the algorithm from \autoref{lem:cpds-algo} to compute $\cmds(G,k)$ in $\poly(n)$ time. Otherwise, we know that each protrusion consists of at most $2^k$ vertices. Applying \autoref{lem:cpds-protrusionred} to all protrusions gives us a graph $G'$ with $\bO(k) + \bO(k) \cdot \bO(k^2) = \bO(k^3)$ vertices and an integer $k' = k + \bO(k) \cdot \bO(k^2) = \bO(k^3)$ such that $\cmds(G,k) = \cmds(G',k')$.
\end{proof}

\section{Conclusion} \label{sec:conclusion}
We introduced a new model of kernelization for counting problems: a polynomial-time preprocessing algorithm that either outputs the desired count, or reduces to a provably small instance with the same answer. We showed that for counting the number of minimum solutions of size at most~$k$, a reduction to a graph of size~$\poly(k)$ exists for two classic problems.

We believe that the new viewpoint on counting kernelization facilitates a general theory that can be explored for many problems beyond the ones considered here. By following the textbook proof~\cite[Lemma 2.2]{course-book} that a decidable parameterized decision problem is fixed-parameter tractable if and only if it admits a kernel (of potentially exponential size), it is easy to show the following equivalence between fixed-parameter tractability of counting problems and our notion of counting kernelization.

\begin{lemma}
Let~$\mathcal{P} \colon \Sigma^* \times \mathbb{N} \to \mathbb{N}$ be a computable function for some finite alphabet~$\Sigma$. Then the following two statements are equivalent:
\begin{enumerate}
    \item There is a computable function~$f \colon \mathbb{N} \to \mathbb{N}$ and an algorithm that, given an input~$(x,k) \in \Sigma^* \times \mathbb{N}$, outputs~$\mathcal{P}(x,k)$ in time~$f(k) \cdot |x|^{\bO(1)}$.
    \item There is a computable function~$f \colon \mathbb{N} \to \mathbb{N}$ and a polynomial-time algorithm that, given~$(x,k) \in \Sigma^* \times \mathbb{N}$, either:
    \begin{enumerate}
        \item outputs~$\mathcal{P}(x,k)$, or
        \item outputs~$(x',k') \in \Sigma^* \times \mathbb{N}$ satisfying~$|x'|, k' \leq f(k)$ and~$\mathcal{P}(x,k) = \mathcal{P}(x',k')$.
    \end{enumerate}
\end{enumerate}
\end{lemma}

Hence our view of counting kernelization is generic enough to capture all fixed-parameter tractable counting problems. Determining which counting problems have a \emph{polynomial-size} kernel remains an interesting challenge.

In this work, we focused on counting \emph{minimum-size} solutions (of size at most~$k$). Apart from being of practical interest in several applications, this facilitates several steps in the design and analysis of our preprocessing algorithms. At present, we do not know whether the two considered problems have polynomial-size kernels when counting the number of inclusion-minimal solutions of size exactly~$k$, or the number of (not necessarily minimal or minimum) solutions of size exactly~$k$; we leave this investigation to future work. To see the importance of the distinction, observe that the number of \emph{minimum} vertex covers of size at most~$k$ in a graph is bounded by~$2^k$ (since the standard 2-way branching discovers all of them) and the Buss kernel preserves their count. But the total number of vertex covers of size at most~$k$ cannot be bounded in terms of~$k$ in general.

For both problems we investigated, our preprocessing step effectively consists of reducing to an equivalent instance composed of a small core along with~$\poly(k)$ simply structured parts, followed by replacing each such part by a small problem-specific gadget. In the world of decision problems, the theory of protrusion replacement~\cite{BodlaenderFLPST16,Fomin10} gives a generic way of replacing such simply-structured parts by gadgets. Similarly, the condenser-extractor framework~\cite{KimST18,Thilikos21} can be applied to generic problems as long as they can be captured in a certain type of logic. This leads to the question of whether, in our model of counting kernelization, the design of the gadgets can be automated. Can a notion of meta-kernelization be developed for counting problems?

\bibliography{references}

\appendix 
\clearpage
\section{An FPT algorithm for \#minFVS} \label{app:fpt-cfvs-algo}
\cfvsFPTAlgo*
\begin{proof}
The pseudocode of an algorithm to solve the disjoint minimum feedback vertex set counting problem is given in \autoref{alg:dj-cfvs} and is based on the algorithm by Cao et al~\cite{fvs-new-measure-cao}. In fact, our algorithm solves a slightly more general version of the problem, where the vertices of the graph $G$ are weighted by a function $w \colon V(G) \to \mathbb{N}$. In case $G$ has no FVS of size at most $k$ that is disjoint from $W$, then \textsc{\#DJ-FVS}($G$, $w$, $W$, $k$) will output a pair $(a, b)$ with $a=\infty$ and $b=0$. Otherwise, $a$ will be the size of a minimum FVS disjoint of $W$ and  
\[b = \sum_{\substack{|S| = a,\\S \in \mathrm{FVS}(G),\\ S \cap W = \emptyset }} \prod_{v \in S} w(v).\]
We refer to this value as the \emph{weighted disjoint minimum FVS sum} of $G$. The weight of a vertex $v$ essentially models the number of distinct alternatives there are for a vertex $v$ that, in the context of choosing them for a minimum FVS of $G$, achieve the same result. When we assign a weight of one to each vertex of a graph $G$, then the weighted minimum FVS sum of $G$ is equal to the number of minimum feedback vertex sets of $G$.

In the pseudocode we make use of binary operator $\oplus$, which is defined to work on pairs as follows:
\[ (a_1, b_1) \oplus (a_2, b_2) = \begin{cases}
    (a_1, b_1) & \mbox{if } a_1 < a_2\\
    (a_2, b_2) & \mbox{if } a_1 > a_2\\ 
    (a_1, b_1 + b_2) & \mbox{if } a_1 = a_2\\ 
\end{cases}\]
For the operators $+$ and $\cdot$ we assume element-wise functionality when applied to pairs, i.e. $(a_1, b_1) + (a_2, b_2) = (a_1 + a_2, b_1 + b_2)$. Furthermore, for a weight function $w$ of $G$ and $X \subseteq V(G)$, we use $w\vert_X$ to denote the restriction of $w$ to $X$. 

We shall now explain how the \textsc{\#DJ-FVS} algorithm works by going over the pseudocode in \Cref{alg:dj-cfvs}. It makes use of a branching strategy with measure function $\ell+k$, where $\ell$ is the number of connected components of $G[W]$. 

Lines 1-3 are the base cases of the algorithm. In lines 4-5 we remove vertices of degree at most one, which corresponds to (R2). In lines 6-7 we contract the chains of $G$ existing in $G-W$ one edge at a time. Note for line 8 that $H$ is a forest since $W$ is an FVS of $G$. For lines 9-10, if a vertex $v$ would form a cycle with $W$ it should be contained in all solutions so we recurse on this choice. In lines 11-14, vertex $v$ has at least two neighbors in $W$, but since $G[W \cup \{v\}]$ does not form a cycle these neighbors must be in different connected components of $G[W]$. Thus we branch on $v$ not being in a solution, which corresponds with adding $v$ to $W$, and on $v$ being in a solution by removing it from the graph. In the former branch $\ell$ decreases, while in the latter $k$ decreases. 

In line 15, we choose a vertex $v \in V(H)$ that is not a leaf of tree $H$ such that at most one of its neighbors in $H$ is not a leaf of $H$.
Note that such a vertex always exists for a tree that does not consist of only leaves. Furthermore, at this point of the algorithm no tree of $H$ can consist of only leaves. If a tree of $H$ consists of a single leaf, then either it has degree one in $G$, but then lines 4-5 would have gotten rid of it, or it has degree at least two, in which case the if condition on line 9 or the if condition on line 11 would have been satisfied. If a tree of $H$ consists of two leaves, then either both have degree two in which case the edge between them would have been contracted by lines 6-7, or one of the if statements of line 9 or line 11 would have been applicable to one of the two leaves.

First consider the case where $v$ has one neighbor in $W$ (lines 16-22). In that case we pick $c$ in line 16 to be a child of $v$. For these two vertices, we branch over all possible combinations of whether or not they should be in a solution, while realizing that a minimum FVS that contains $v$ can never contain $c$. Note that if $G[W \cup \{v, c\}]$ does not form a cycle, both $v$ and $c$ must have a neighbor in a different connected component of $W$ so the branch in line 19 decreases $\ell$. The branches on line 20 and 21 decrease $k$ while not increasing $\ell$. 

Finally, we have the case that $v$ has no neighbors in $W$ (lines 23-31). That must mean that $v$ has at least two children, which are all leaves by how we chose $v$. If $v$ would have had only one child, then the edge between $v$ and the child would have been contracted in lines 6-7. Thus we let $c_1$ and $c_2$ be two distinct children of $v$. Again, we branch over all viable combinations of how these three vertices can be part of solutions. The logic here is similar to that of the previous case. Here as well, in all branches, the measure $k + \ell$ decreases. 

Since the algorithm branches in at most five directions and the time per iteration is polynomial in $n$, the runtime of the disjoint algorithm becomes $5^{k+\ell} \cdot n^{O(1)}$.
We can use \autoref{alg:dj-cfvs} to compute $\cmfvs(G,k)$ by taking an FVS $Z$ of $G$ and running the disjoint algorithm for all subsets of $Z$, simulating the ways an FVS can intersects $Z$. The pseudocode for this compression algorithm can be seen in \autoref{alg:compr-cfvs}. 
To get an FVS of $G$ of size at most $k$, we can simply run one of the existing algorithms designed for this, for example the one by Cao et al.~\cite{fvs-new-measure-cao} which runs in time $\bO(3.83^k) \cdot \poly (n)$. If this reports that no such FVS exists, we output $\cmfvs(G,k) = 0$. Otherwise, we use the found FVS for the compression algorithm.
Using the fact that the FVS used by the compression algorithm is of size at most $k$, we can bound the runtime of the complete algorithm to compute $\cmfvs(G,k)$ at $26^k \cdot n^{O(1)}$.
\end{proof}
\begin{algorithm}[h]
\caption{\textsc{\#FVS-compression}($G$, $k$, $Z$)}
\hspace*{\algorithmicindent}\textbf{Input:} Graph $G$, integer $k$, FVS $Z$ of $G$ of size at most $k$ \\
\hspace*{\algorithmicindent}\textbf{Output:} A pair $(a,b)$ with $a$ being the feedback vertex number of $G$ and \\
\hspace*{\algorithmicindent}$b = \cmfvs(G,k)$
\begin{algorithmic}[1]  
\State $s = (\infty, 0)$
\For{$X_Z \subseteq Z$}
\State $s' = (|X_Z|, 0) + $\Call{\#DJ-FVS}{$G - X_Z$, $w$, $Z \setminus X_Z$, $k-|X_Z|$} with $\forall v \in V(G - X_Z): w(v) = 1$
\State $s = s \oplus s'$
\EndFor
\State \Return s
\end{algorithmic}
\label{alg:compr-cfvs}
\end{algorithm}

\begin{algorithm}[h]
\caption{\textsc{\#DJ-FVS}($G$, $w$, $W$, $k$)}
    \hspace*{\algorithmicindent}\textbf{Input:} Graph $G$, weight function $w \colon V(G) \to \mathbb{N}$, FVS $W$ of $G$, integer $k$  \\
    \hspace*{\algorithmicindent}\textbf{Parameter:} $k + \ell$, with $\ell = \#$ components of $G[W]$ \\
    \hspace*{\algorithmicindent}\textbf{Output:} A pair $(a, b)$ where $a$ is the size of a minimum FVS $S$ of $G$ such that\\
    \hspace*{\algorithmicindent}$S \cap W = \emptyset$ and $b$ is the weighted disjoint minimum FVS sum of $G$. If no such FVS\\
    \hspace*{\algorithmicindent}of size at most $k$ exists, then $(a, b) = (\infty, 0)$.
\begin{algorithmic}[1]  
\If{$k < 0$} \Return $(\infty,0)$ \EndIf
\If{$G[W]$ has a cycle} \Return $(\infty,0)$  \EndIf
\If{$G-W$ is empty} \Return $(0,1)$  \EndIf
\If{$\exists v \in V(G-W)$: $\mathit{deg}_G(v) \leq 1$} 
    \State \Return \Call{\#DJ-FVS}{$G-v$, $w\vert_{V(G) \setminus \{v\}}$ ,$W$,$k$} 
\EndIf
\If{$\exists \{v,u\} \in E(G):$ $\deg_G(v) = \deg_G(u) = 2$ and $v,u \notin W$}
\State \Return \Call{\#DJ-FVS}{$G'$, $w'$ $W$, $k$}, where $G'$ is $G$ with edge $\{v,u\}$ contracted to a single vertex $s$ and $w'$ is the weight function $w\vert_{V(G')}$ with $w'(s) = w(v) + w(u)$.
\EndIf
\State Let $H$ be forest $G-W$.
\If{$\exists v \in V(H)$: $G[W \cup \{v\}]$ has a cycle} 
    \State \Return $(1,0) + (1,w(v)) \cdot $ \Call{\#DJ-FVS}{$G-v$, $w\vert_{V(G)\setminus \{v\}}$, $W$, $k-1$} 
\EndIf
\If{$\exists v \in V(H)$: $|N_{G}(v) \cap W| \geq 2$}
    \State $X_0 = $\Call{\#DJ-FVS}{$G$, $w$, $W \cup \{v\}$, $k$}
    \State $X_1 = (1,0) + (1,w(v)) \cdot$ \Call{\#DJ-FVS}{$G-v$, $w\vert_{V(G)\setminus \{v\}}$, $W$, $k-1$}
    \State \Return $X_0 \oplus X_1$
\EndIf
\State Let $v \in V(H)$ be a vertex that is not a leaf of $H$ such that at most one vertex in $N_H(v)$ is not a leaf of $H$.
\If{$|N_G(v) \cap W| = 1$}
    \State Pick $c \in V(H)$ such that $N_{G}(c) = \{v, x\}$ for some $x \in W$
    \If{$G[W \cup \{v, c\}]$ forms a cycle} $X_{00} = (\infty,0)$ 
    \Else\space $X_{00} = $ \Call{\#DJ-FVS}{$G$, $w$, $W \cup \{v, c\}$, $k$}
    \EndIf
    \State $X_{10} = (1,0) + (1, w(v)) \cdot$ \Call{\#DJ-FVS}{$G-v$, $w\vert_{V(G)\setminus \{v\}}$, $W$, $k-1$}
    \State $X_{01} = (1,0) +  (1, w(c)) \cdot$ \Call{\#DJ-FVS}{$G-c$, $w\vert_{V(G)\setminus \{c\}}$,  $W \cup \{v\}$, $k-1$}
    \State \Return $X_{00} \oplus X_{10} \oplus X_{01}$
\EndIf
\If{$|N_{G}(v) \cap W| = 0$}
    \State Pick $c_1, c_2 \in V(H), c_1 \neq c_2$ such that $N_{G}(c_1) = \{v, x\}$ and $N_{G}(c_2) = \{v, y\}$ for some $x,y \in W$
    \If{$G[W \cup \{v, c_1, c_2\}]$ forms a cycle} $X_{000} = (\infty,0)$ 
    \Else\space $X_{000} = $ \Call{\#DJ-FVS}{$G$, $w$, $W \cup \{v, c_1, c_2\}$, $k$}
    \EndIf
    \State $X_{100} = (1,0) + (1,w(v)) \cdot$ \Call{\#DJ-FVS}{$G-v$, $w\vert_{V(G)\setminus \{v\}}$, $W$, $k-1$}
    \State $X_{010} = (1,0) + (1,w(c_1)) \cdot$ \Call{\#DJ-FVS}{$G-c_1$, $w\vert_{V(G)\setminus \{c_1\}}$, $W \cup \{v, c_2\}$, $k-1$}
    \State $X_{001} = (1,0) + (1, w(c_2)) \cdot$ \Call{\#DJ-FVS}{$G-c_2$, $w\vert_{V(G)\setminus \{c_2\}}$,  $W \cup \{v, c_1\}$, $k-1$}
    \State $X_{011} = (2,0) + (1,w(c_1) \cdot w(c_2)) \cdot$ \Call{\#DJ-FVS}{$G-\{c_1,c_2\}$, $w\vert_{V(G)\setminus \{c_1,c_2\}}$, $W \cup \{v\}$, $k-2$}
    \State \Return $X_{000} \oplus X_{100} \oplus X_{010} \oplus X_{001} \oplus X_{011}$
\EndIf
\end{algorithmic}
\label{alg:dj-cfvs}
\end{algorithm}

\end{document}